\documentclass[12pt]{article}

\usepackage{amsmath}
\usepackage{amssymb}

\usepackage[dvips]{graphicx}
\usepackage[usenames,dvipsnames,svgnames]{xcolor}

\newtheorem{theorem}{Theorem}[section]
\newtheorem{definition}[theorem]{Definition}

\newtheorem{corollary}[theorem]{Corollary}

\newenvironment{proof}
{\noindent
{\bf Proof.}}
{\hfill $\square$\medskip

}

\newcommand{\Id}{\operatorname{Id}}
\newcommand{\R}{\mathbb{R}}

\newcommand{\tr}{^{\rm t}}
\newcommand{\sprod}[1]{\langle #1 \rangle}
\newcommand{\br}{{\mathbf r}}
\newcommand{\bs}{{\mathbf s}}
\newcommand{\bx}{{\mathbf x}}
\newcommand{\bt}{{\mathbf t}}
\newcommand{\bc}{{\mathbf c}}
\newcommand{\bT}{{\mathbf T}}
\newcommand{\bp}{{\mathbf p}}
\newcommand{\bn}{{\mathbf n}}

\newcommand{\bZero}{{\mathbf 0}}
\newcommand{\bnu}{{\pmb \nu}}
\newcommand{\bxi}{{\pmb \xi}}
\usepackage{hyperref}

\usepackage{xfrac}
\usepackage{xspace} 
\usepackage{cleveref}
\crefformat{figure}{Figure~#2#1#3} 
\Crefformat{figure}{Figure~#2#1#3} 
\crefformat{equation}{(#2#1#3)} 
\Crefformat{equation}{(#2#1#3)} 
\labelcrefformat{equation}{(#2#1#3)}  
\crefformat{section}{Section~#2#1#3} 
\Crefformat{section}{Section~#2#1#3} 
\usepackage{tikz}
\usepackage{tikz-3dplot}
\usepackage{pgfplots}
\usepackage{standalone}   
\usepgfplotslibrary{groupplots}
\usetikzlibrary{external} 
\tikzsetexternalprefix{tmp/}
\tikzexternaldisable

\graphicspath{{pics/}{../pics/}{tmp/}{.}}

\pgfplotsset{compat=1.15}
\def\pgfplotfontsizetitle{\normalsize}
\def\pgfplotfontsize{\normalsize}

\def\tikzfontsizetiny{\scriptsize}

\pgfplotsset{
  mystyle/.style ={%
    grid = major,
    every tick label/.append style={font=\pgfplotfontsize},
    every axis label/.append style={font=\pgfplotfontsize},
    legend style={%
        font=\pgfplotfontsize,%
        draw=none,%
        /tikz/every even column/.append style={column sep=8pt}%
    },
    label style={font=\pgfplotfontsize},
    title style={font=\pgfplotfontsizetitle},
    /pgf/number format/set thousands separator = {}, 
  }
}

\pgfplotscreateplotcyclelist{mycyclelisthemiMore}{
    blue,no marks\\%
    red,no marks\\%
    brown!60!black,no marks\\%
    black,no marks\\%
    blue,loosely dashed,every mark/.append style={solid},mark size={2.5},mark=square\\%
    red,loosely dashed,every mark/.append style={solid},mark size={2.5},mark=o\\%
    brown!60!black,loosely dashed,every mark/.append style={solid},mark size={3},mark=diamond\\%
    black,loosely dashed,every mark/.append style={solid},mark size={3},mark=triangle\\%
    blue,dotted,every mark/.append style={solid},mark size={1.5},mark=square*\\%
    red,dotted,every mark/.append style={solid},mark size={1.5},mark=*\\%
    brown!60!black,dotted,every mark/.append style={solid},mark size={1.5},mark=diamond*\\%
    black,dotted,every mark/.append style={solid},mark size={1.5},mark=triangle*\\%
}
\pgfplotscreateplotcyclelist{mycyclelisthemi}{
    blue,loosely dashed,every mark/.append style={solid},mark size={2.5},mark=square\\%
    red,loosely dashed,every mark/.append style={solid},mark size={2.5},mark=o\\%
    brown!60!black,loosely dashed,every mark/.append style={solid},mark size={3},mark=diamond\\%
    blue,dotted,every mark/.append style={solid},mark size={1.5},mark=square*\\%
    red,dotted,every mark/.append style={solid},mark size={1.5},mark=*\\%
    brown!60!black,dotted,every mark/.append style={solid},mark size={1.5},mark=diamond*\\%
}

\newcommand{\upperSlopeTriangle}[4] 	
{
	\addplot[forget plot, domain=#3:#4,color=black,samples=2]{  #2 / (x^#1) } node (A1) [pos=1] {}; 
	\addplot[forget plot, domain=#3:#4,color=black,samples=2]{   #2  / (#3^#1)} node (A2) [pos=1] {} node [anchor=south,pos=0.5,black] {\tikzfontsizetiny $1$};
	\draw[color=black] (A1.center) -- (A2.center) node [anchor=west,pos=0.5,black] {\tikzfontsizetiny #1};
}

\newcommand{\lowerSlopeTriangle}[4] 	
{
	\addplot[forget plot, domain=#3:#4,color=black,samples=2]{  #2 / (x^#1) } node (A1) [pos=0] {}; 
	\addplot[forget plot, domain=#3:#4,color=black,samples=2]{   #2  / (#4^#1)} node (A2) [pos=0] {} node [anchor=north,pos=0.5,black] {\tikzfontsizetiny $1$};
	\draw[color=black] (A1.center) -- (A2.center) node [anchor=east,pos=0.5,black] {\tikzfontsizetiny #1};
}

\newcommand{\myaddgraphic}[5]
{
 \node[anchor=south west,inner sep=0] (image) {\phantom{\includegraphics[#2]{#1}}};
  \begin{scope}[x={(image.south east)},y={(image.north west)}]
      
      \begin{scope}
          
          #5
          
          \node[anchor=south west,inner sep=0] {\includegraphics[#2]{#1}};
      \end{scope} 
      
      #4
      
      \pgfmathparse{int(#3)} \let\gridIndicator\pgfmathresult
      \ifthenelse{ \gridIndicator = 1 }
      {
          \draw[help lines,xstep=.1,ystep=.1] (0,0) grid (1.001,1.001);
          \foreach \x in {1,...,9} { \node [anchor=north] at (\x/10,0) {\x};}
          \foreach \y in {1,...,9} { \node [anchor=east] at (0,\y/10) {\y};}
      }{}
      
  \end{scope}    
}

\newcommand{\fitByL}{standard\xspace}
\newcommand{\fitByRC}{constrained\xspace} 
\newcommand{\splineFit}{\mathbf{x}}
\newcommand{\splineDegree}{n}

\newcommand{\splineU}{u}
\newcommand{\splineV}{v}

\newcommand{\tikzfig}[4]{
\def\tkzscale{#2}
\def\dirtikz{tikz} 
\def\dirdata{data} 
\centering
\includegraphics{#1}
\caption{#3}
\label{#4}
}

\newcommand{\myHighlight}[1]{\color{black}#1\color{black}\xspace}

\begin{document}
\title{Surface Patches with Rounded Corners}

\author{Benjamin Marussig and Ulrich Reif}

\date{\today}
\maketitle

\begin{abstract} 
We analyze surface patches with a corner that is rounded 
in the sense that the partial derivatives at that point are
antiparallel. Sufficient conditions for $G^1$ smoothness are 
given, which, up to a certain degenerate case, are also 
necessary. Further, we investigate curvature integrability
and present examples.
\end{abstract}

\section{Introduction}
Surface parametrizations of form 
$\bx : [a,b] \times [c,d] \to \R^3$,
such as B\'ezier patches or NURBS surfaces,
are frequently used in geometric modeling. The regularity of $\bx$ 
in the sense that $\det D\bx \neq 0$ is a standard assumption
to guarantee the geometric smoothness of the trace of $\bx$, 
and also, most analytic tools from differential geometry rely on that assumption.
However, regularity implies that the shape of the parametrized
surface is necessarily four-sided. Moreover, the parametrically 
smooth contact of $n \neq 4$ such patches sharing a vertex in 
a composite model is impossible. Methods to overcome these 
restrictions include trimming \cite{marussig2017a}
and the concept of geometric continuity \cite{Peters02}. Another approach 
is based on deliberately
dropping regularity at isolated spots of the surface. This increases 
flexibility, but special care has to be taken that the resulting 
surfaces are geometrically smooth of order $G^k$ in a vicinity of 
the singularity in the sense that, locally, there exists a regular 
reparametrization of class $C^k$.

There exist different types of patches with singularities:
First, certain partial derivatives of $\bx$ can be set to zero
at a corner of the domain. In \cite{Reif97} and \cite{Bohl97}, 
conditions for $C^1$- and $C^2$-smoothness of such parametrizations 
are derived, see also \cite{Sederberg11}. These constructions are 
useful for the parametrically smooth contact of $n \neq 4$ patches 
meeting at a point.

Second, whole edges of the domain can be requested to collapse
to single points in the image of $\bx$, see for instance 
\cite{Yan14}.
This facilitates the representation of three- and also two-sided 
shapes. \myHighlight{In \cite{Speleers20211a}, an extraction matrix is utilized to construct $C^1$ smooth splines on such shapes, allowing smooth single-patch parametrizations of ellipsoids.} 

Third, the first partial derivatives of $\bx$ at certain corners of 
the domain may be {\em antiparallel}. This 
means that the edges sharing such corners are mapped to 
curves meeting at straight angles. Again, this construction 
admits the representation of three- or two-sided shapes. But it is
equally possible to parametrize, for instance, a hemispherical shape
by a single map $\bx : [0,1]^2 \to \R^3$, see \cref{sec:hemisphere}.

Surface patches with such {\em rounded corners} appear for instance in watertight Boolean operations presented in \cite{Urick19}. 
In this approach, turning points of trimming curves are utilized to define a layout of surface patches for constructing non-trimmed watertight boundary representations of volumes in $\R^3$.
In the resulting model, a turning point becomes a corner of the surfaces obtained, and if the initial trimming curve is smooth at this point, this corner is prone to be a rounded one.
Other applications can be found in isogeometric analysis, 
where patches with rounded corners may be convenient to 
parametrize the physical domain, see for instance \cite{Takacs12}.
The latter reference addresses aspects of Sobolev regularity, while a 
specific analysis of geometric smoothness seems to be unknown in the 
literature. This paper aims at filling that gap.

In \cref{sec:analysis}, we define surface patches with rounded corners 
and provide sufficient conditions for $C^1$-smoothness. These 
conditions are shown to be almost necessary in the sense that 
only cases with degeneracies of higher order are left undecided.
Further, we show that the principal curvatures of such $C^1$-patches 
are square-integrable. 
In \cref{sec:experiments}, we present two examples featuring the use of surface 
patches with rounded corners in applications.

\section{Analysis of rounded corners}
\label{sec:analysis}
To simplify notation, we consider surface patches with 
domain $[0,\myHighlight{H}]^2$ and analyze their behavior in the vicinity
of the vertex $(0,0)$. The generalization to an arbitrary corner 
of an arbitrary rectangle is straightforward.

\begin{definition}
\label{def:rounded}
Let 
\begin{equation}
\label{eq:s}
 \bx : [0,\myHighlight{H}]^2 \to \R^3
 ,\quad 
 \bx(u,v) = \sum_{j+k =0}^2 \frac{u^j v^k}{j!k!} \, \bxi_{j,k} 
 + O\bigl((u+v)^3\bigr)
 ,
\end{equation}
be a three times differentiable surface patch. It has a 
{\em rounded corner} 
at $(u,v) = (0,0)$ if
the following conditions are satisfied:
\begin{itemize}
\item {\bf Antiparallelism.}
There exists a unit vector $\bt \in \R^3$ and factors $\mu,\lambda>0$ such 
that 
\[
  \bxi_{1,0} = \lambda \bt
  ,\quad 
  \bxi_{0,1} = -\mu \bt
  .
\]
\item {\bf Coplanarity.} The three vectors
\[
  \br := \mu \bxi_{2,0} + \lambda \bxi_{1,1}
  ,\quad 
  \bs := \lambda \bxi_{0,2}+\mu \bxi_{1,1}
  ,
\]
and $\bt$ are linearly dependent.
\item {\bf Onesidedness.} \myHighlight{The three vectors satisfy} 
\[
 \sprod{\bt \times \br,\bt \times \bs} > 0 
  .
\]
\end{itemize}

\end{definition}

Throughout, we will assume that the size
$H$ of the domain is chosen so small that properties of the surface $\bx$ away from 
the origin can be discarded. \myHighlight{Furthermore, we define 
\[
  h := \max\{u,v\}.
  \]}

Antiparallelism of the partial derivatives $\bxi_{1,0},\bxi_{0,1}$ causes a 
loss of regularity of $\bx$ at the origin. This means that \myHighlight{geometric smoothness of the trace} cannot be taken for granted at that point, 
despite the smoothness of the parametrization. However, we are going to demonstrate
that coplanarity together with onesidedness guarantees that
the patch $\bx$ is $G^1$, meaning that there exists a 
regular $C^1$-parametrization of the trace of $\bx$ near the rounded corner.
We will also show that the $G^1$-property is lost if the vectors
$\br,\bs,\bt$ are linearly independent, or if the quadruple product 
is negative, leaving only the particular case 
$\sprod{\bt \times \br,\bt \times \bs} = 0$ undecided.

Onesidedness implies that the {\em limit normal}
\[
 \bn := 
 \frac{\bt \times \br}{\|\bt \times \br\|}
\]
is well defined. Together, $\bt$ and the {\em cross vector}
\[
 \bc := \bn \times \bt
\]
span the {\em limit tangent space}
\[
 \bT := \{\alpha \bt + \beta \bc : 
 (\alpha,\beta)\in \R^2\}
 ,
\]
and we observe that $\br,\bs \in \bT$. Denoting 
the {\em cross components} of $\br$ and $\bs$ by
\[
  \varrho := \sprod{\bc,\br}
  ,\quad 
  \sigma := \sprod{\bc,\bs}
  ,
\]
respectively, \myHighlight{we have} 
\begin{equation}
\label{eq:txr}
  \bt \times \br = \varrho \bn
  \quad \text{and}\quad \bt \times \bs = \sigma \bn
  .
\end{equation}
Onesidedness yields 
$\sprod{\bt \times \br,\bt \times \bs} = \varrho\sigma >0 $ so that
\begin{equation}
\label{eq:positive}
 \varrho = \|\bt \times \br\|>0
 \quad \text{and}\quad
 \sigma>0
 .
\end{equation}
Geometrically speaking, positivity of the cross components means that the 
vectors $\br$ and $\bs$ lie on the same side of the vector $\bt$ within the 
plane $\bT$, which accounts for the name of the third property in 
Definition~\ref{def:rounded}.

The following theorem shows that $\bn$ is in fact the limit 
of normal vectors 
\[
 \bnu := \frac{\bx_u \times \bx_v}{\|\bx_u \times \bx_v \|}
\]
at the rounded corner.
\begin{theorem}[Normal continuity.]
\label{thm:normcont}
The surface patch $\bx$ with a rounded corner according to 
Definition~\ref{def:rounded} is normal continuous at $(0,0)$ 
with
\[
 \bn = \lim_{(u,v)\to (0,0)} \bnu(u,v)
 .
\]
\end{theorem}
\begin{proof}
The partial derivatives
of $\bx$ are
\begin{equation}
\label{eq:partial}
 \bx_u = \lambda\,\bt + u\,\bxi_{2,0} +  v \,\bxi_{1,1} + O(h^2)
 ,\quad 
 \bx_v = -\mu\,\bt + v\,\bxi_{0,2} + u\,\bxi_{1,1} + O(h^2)
 .
\end{equation}
Their cross product is
\begin{align}
\notag
 \bx_u \times \bx_v &=
 \bt \times \bigl(
 u(\mu\, \bxi_{2,0}+\lambda\, \bxi_{1,1}) +
 v(\mu\, \bxi_{1,1}+\lambda\, \bxi_{0,2})\bigr)
 + O(h^2)\\
 \label{eq:cross}
 &=
 \bt \times ( u\, \br + v\, \bs ) + O(h^2) = 
 u\, \bt \times \br + v\, \bt \times \bs + O(h^2)\\
 &=
 \label{eq:cross_n}
 (\varrho u + \sigma v)\, \bn + O(h^2) 
 ,
\end{align}
where we used \eqref{eq:txr} to derive the last equality.
By positivity of the cross components according to \eqref{eq:positive}, 
the reciprocal of the first factor in \eqref{eq:cross_n} is bounded by 
\begin{equation}
\label{eq:lowerbnd}
 0 < \frac{1}{\varrho u + \sigma v} \le 
 \frac{1}{\min\{\varrho,\sigma\}h}
 = O(1/h)
 ,\quad 
 (u,v) \in [0,\myHighlight{H}]^2 \setminus \{(0,0)\}
 .
\end{equation}
Hence,
\[
 \frac{1}{\|\bx_u \times \bx_v\|}
 = \frac{1}{(\varrho u + \sigma v) + O(h^2)}
 = \frac{1}{(\varrho u + \sigma v)(1+O(h))} = O(1/h)
 ,
\]
and convergence of normal vectors follows from 
\[
 \bnu = 
 \frac{\bx_u \times \bx_v}{\|\bx_u\times \bx_v\|}
 = \frac{(\varrho u + \sigma v)\bn}{(\varrho u + \sigma v)(1+O(h))} 
 + O(h)
 = \bn + O(h)
 .
\]
\end{proof}
The next theorem clarifies that coplanarity and onesidedness are 
essential for normal continuity when antiparallelism is assumed.
\begin{theorem}[Normal discontinuity.]
\label{thm:normdiscont}
Let $\bx$ be a surface patch with antiparallel partial derivatives
$\bxi_{1,0},\bxi_{0,1}$ at the origin, as in Definition~\ref{def:rounded}. 
If the vectors $\br,\bs,\bt$ are linearly independent, or if
\[
 \sprod{\bt \times \br,\bt \times \bs} < 0
 ,
\]
then $\bx$ is not normal continuous.
\end{theorem}
\begin{proof}
First, let us assume that the vectors $\br,\bs,\bt$ are linearly independent.
Then also the vectors $\bn_1 := \bt \times \br$ and 
$\bn_2 := \bt \times \bs$ are linearly independent since 
$\bn_1 \times \bn_2 = \det [\br,\bs,\bt]\, \bt \neq \bZero$.
According to \eqref{eq:cross}, \myHighlight{we have}
\[
\bx_u \times \bx_v =
 u\, \bt \times \br + v\, \bt \times \bs + O(h^2)
 =
 u \bn_1 + v \bn_2 + O(h^2)
 .
\]
Comparing 
\[
 \bnu(u,0) = \frac{u \bn_1 + O(u^2)}{u\|\bn_1\|(1 + O(u))}
 = \frac{\bn_1}{\|\bn_1\|} + O(u)
\]
and
\[
 \bnu(0,v) = \frac{v\bn_2 + O(v^2)}{v\|\bn_2\|(1 + O(v))}
 = \frac{\bn_2}{\|\bn_2\|} + O(v)
\]
shows that $\bnu$ does not have a \myHighlight{unique} limit at $(0,0)$.

Second, let us assume that the vectors $\br,\bs,\bt$ are coplanar,
and that
$\sprod{\bt \times \br,\bt \times \bs} < 0$. Then we can follow 
the proof of the preceding theorem up to \eqref{eq:cross_n},
\[
  \bx_u \times \bx_v = 
 (\varrho u + \sigma v)\, \bn + O(h^2)
 ,
\]
but now, the factors $\varrho$ and $\sigma$ have opposite sign, 
\[
 \sprod{\bt \times \br,\bt \times \bs} = \varrho \sigma < 0
 .
\]
As before, comparing
\[
 \bnu(u,0) = 
 \frac{\varrho u\, \bn + O(u^2)}{|\varrho|u\, \|\bn\|(1 +O(u))}
 = \bn + O(u)
\]
and
\[
 \bnu(0,v) = 
 \frac{\sigma v\, \bn + O(v^2)}{|\sigma|v\, \|\bn\|(1 +O(v))}
 = -\bn + O(v)
\]
shows that $\bnu$ does not have a limit at $(0,0)$.
\end{proof}
The only case of surface patches with antiparallel partial derivatives
$\bxi_{1,0},\bxi_{0,1}$
not covered by \myHighlight{Theorems}~\ref{thm:normcont} and 
\ref{thm:normdiscont} is that of coplanar vectors $\br,\bs,\bt$
with $\sprod{\bt \times \br,\bt \times \bs} = 0$,
which represents a degeneracy of higher order.

Normal continuity is a relatively weak notion of smoothness since
it does not imply that the trace of the 
given parametrization is a smooth manifold. In particular, local 
self-intersections cannot be excluded. As an example, consider
the surface
\[
 \bx(u,v) = 
 \begin{bmatrix} 
 u^7 - 21 u^5 v^2 + 35u^3 v^4 - 7 u v^6\\
 v^7 - 21 v^5 u^2 + 35v^3 u^4 - 7 v u^6\\
 u^{10} + v^{10}
 \end{bmatrix}
 ,\quad 
 [u,v] \in [0,1]^2
 ,
\]
see \cref{fig:selfintersect} {\em (left)}. It is easily verified by 
inspection that 
\[
 \lim_{(u,v) \to (0,0)} \bnu(u,v) = (0,0,1)\tr
 ,
\]
but the projection of $\bx$ into the $xy$-plane is not injective,
see \cref{fig:selfintersect} {\em (right)}.
\begin{figure}
 \centering
 \includegraphics[height=6cm]{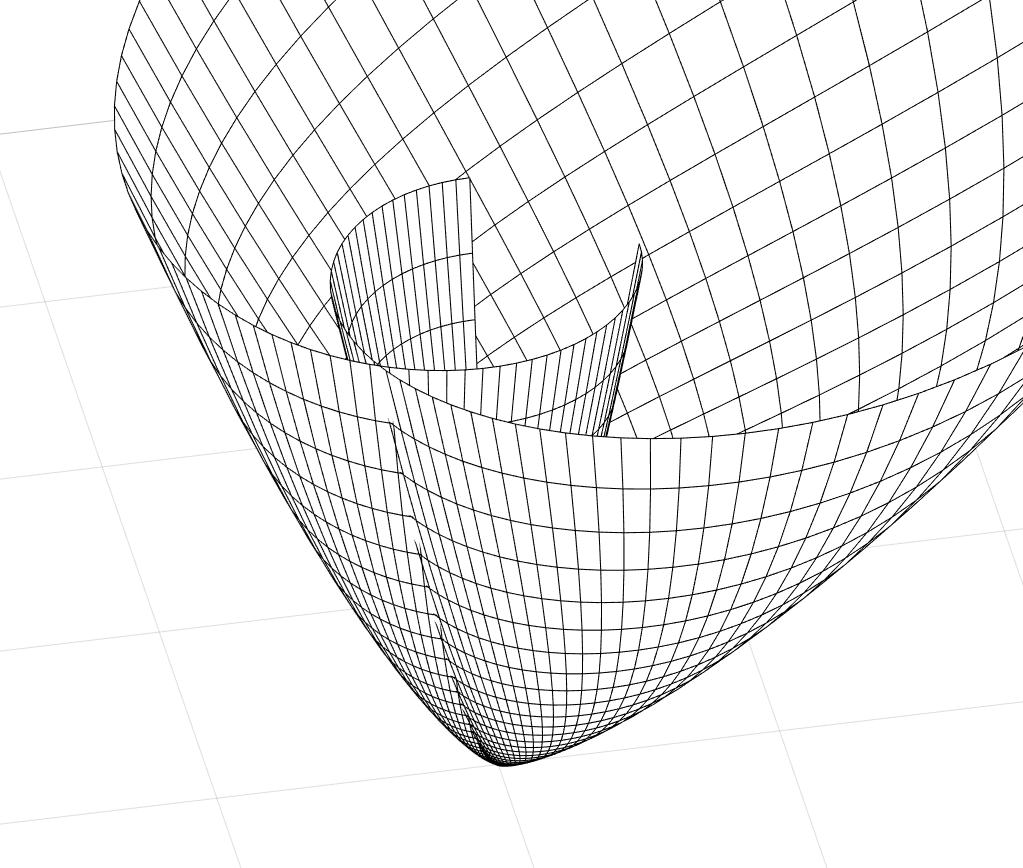}
 \qquad \qquad
 \includegraphics[height=6cm]{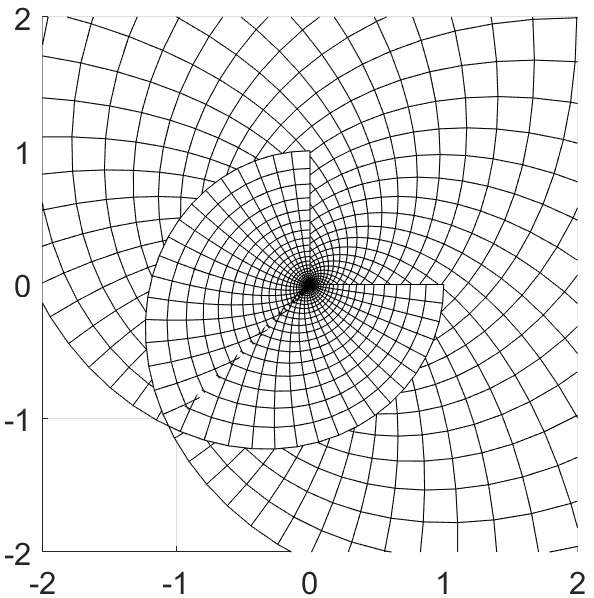}
 \caption{Normal-continuous surface with local self-intersection
 {\em (left)} and projection to tangent plane {\em (right)}.}
 \label{fig:selfintersect}
\end{figure}
The following result states that this cannot happen near 
rounded corners if the conditions of Definition~\ref{def:rounded}
are satisfied.
\begin{theorem}[Single-sheetedness.]
Let $\bx$ be a surface patch with a rounded corner according to
Definition~\ref{def:rounded}.
When restricted to a sufficiently small neighborhood $[0,H]^2$
of the origin, the orthogonal projection 
of $\bx-\bx_{0,0}$ to the limit tangent space $\bT$, \myHighlight{corresponding to the map} 
\[
 \Pi : [0,H]^2 \to \bT
 ,\quad 
 \Pi(\eta) := (\Id - \bn \bn\tr)\,(\bx(\eta)-\bx_{0,0})
 ,
\]
is injective.
\end{theorem}
\begin{proof}
Assume that $\Pi(\eta_0) = \Pi(\eta_1)$
for $\eta_0,\eta_1 \in [0,H]^2$, and let 
\[
 \delta := M^{-1} (\eta_1 - \eta_0)
 ,\quad 
 M :=
 \begin{bmatrix}
 \lambda & \mu \\-\mu  & \lambda 
 \end{bmatrix}
 .
\]
Then the scalar function 
\[
 g(t) := \delta\tr \, [\bt\ \ \bc]\tr\, \Pi(\eta_0 + M \delta t)
\]
has equal values at $t=0$ and $t=1$,
\[
 g(1) - g(0) = \delta\tr \, [\bt\ \ \bc]\tr\, 
 \bigl( \Pi(\eta_1) - \Pi(\eta_0)\bigr) = 0
 .
\]
By the mean value theorem, there exists $\tau \in (0,1)$ such that
\[
 g'(\tau) = \delta\tr \, [\bt\ \ \bc]\tr\,
 (\Id - \bn \bn\tr)\, D\bx(u,v) M \delta = 
 \delta\tr\, \bigl([\bt\ \ \bc]\tr\, D\bx(u,v) M\bigr)\, \delta
 =0
 ,
\]
where $(u,v) = \eta_0 + M\delta \tau \in [0,H]^2\backslash \{(0,0)\}$ 
by convexity of the domain.
\myHighlight{Here, we used that, by definition, both $\bt$ and $\bc$ are perpendicular to $\bn$.}
Denoting the matrix in parentheses by $J$, we obtain 
$g'(\tau) = \delta\tr J \delta$.
With $D\bx = [\bx_u\ \bx_v]$ and \eqref{eq:partial}, a short computation 
yields
\[
 J := 
 \begin{bmatrix}
 \lambda^2 + \nu^2 + O(h) & O(h) \\
 O(h) & \varrho u + \sigma v + O(h^2)
 \end{bmatrix}
 .
\]
We symmetrize the quadratic form by setting $J_s := (J+J\tr)/2$ and use 
\eqref{eq:lowerbnd} again to obtain $g'(\tau) = \delta\tr J_s \delta$ with
\[
 J_s := 
 \begin{bmatrix}
 (\lambda^2 + \nu^2)(1 + O(h)) & O(h) \\
 O(h) & (\varrho u + \sigma v)(1 + O(h))
 \end{bmatrix}
 .
\]
Recalling $0<h = \max\{u,v\} \le H$,
\begin{align*}
 \det J_s &= 
 (\lambda^2 + \nu^2)(\varrho u + \sigma v)(1+O(h))\\
 \operatorname{trace} J_s &=
 (\lambda^2 + \nu^2)(1 + O(h))
\end{align*}
shows that $J_s$ is positive definite provided that $H$ is sufficiently small.
Eventually, $g'(\tau) = \delta\tr J_s \delta = 0$ implies $\delta = 0$ and 
$\eta_0 = \eta_1$, showing that $\Pi$ is injective.
\end{proof}
Together, normal continuity and single-sheetedness imply that 
the parametrization of the patch $\bx$ as a graph over the limit tangent space
is  $C^1$, see \cite[Theorem~2.13]{Peters08}. We state this result as
\begin{corollary}[$C^1$-regularity.]
A surface patch $\bx$ with a rounded corner according to
Definition~\ref{def:rounded} possesses a regular $C^1$-parametrization
in a neighborhood of that corner.
\end{corollary}
The $G^1$-property of surface patches with a rounded corner is a
prerequisite for many design applications. However, also the asymptotic behavior
of the principal curvatures $\kappa_1,\kappa_2$ is significant. In particular, 
square integrability is requested
when such patches shall be used for the Ritz-Galerkin simulation of 4th order 
PDEs, like thin shell equations. The 
following theorem 
settles this issue.
\begin{theorem}[Curvature integrability.]
In a neighborhood of the rounded corner, the principal curvatures 
$\kappa_{1,2}$ of a surface patch $\bx$ 
according to Definition~\ref{def:rounded}
are almost in $L^3$ in the sense that  
\[
 \int_{[0,H]^2} |\kappa_{1,2}|^p \, d\mu < \infty
 \quad \text{for any}\quad 
 p\in [1,3)
 ,
\]
where $d\mu$ denotes the surface element of $\bx$ 
and $H>0$ is chosen sufficiently small.
\end{theorem}
\begin{proof}
The first fundamental form of $\bx$ is
\[
 G :=
 \begin{bmatrix}
  \sprod{\bx_u,\bx_u} & \sprod{\bx_u,\bx_u} \\
  \sprod{\bx_u,\bx_v} & \sprod{\bx_v,\bx_v}
 \end{bmatrix}
 =
 \begin{bmatrix}
 \lambda^2 & -\lambda\mu \\
 -\lambda\mu & \mu^2
 \end{bmatrix} + O(h)
 .
\]
By \eqref{eq:cross_n} and \eqref{eq:lowerbnd}, its inverse is given by 
\[
 G^{-1} = 
 \frac{1}{\|\bx_u \times \bx_v\|^2}\,
 \begin{bmatrix}
 \mu^2 & \lambda\mu \\
 \lambda\mu & \lambda^2
 \end{bmatrix} + O(1/h)
 .
\]
With the second fundamental form
\[
 B :=
 \begin{bmatrix}
  \sprod{\bx_{uu},\bnu} & \sprod{\bx_{uv},\bnu} \\
  \sprod{\bx_{uv},\bnu} & \sprod{\bx_{vv},\bnu}
 \end{bmatrix}
 =
 \begin{bmatrix} 
 \sprod{\bxi_{2,0},\bn} & \sprod{\bxi_{1,1},\bn} \\
 \sprod{\bxi_{1,1},\bn} & \sprod{\bxi_{0,2},\bn}
 \end{bmatrix}
 + O(h)
 ,
\]
we obtain the shape operator
\[
  S := G^{-1} B =
 \frac{1}{\|\bx_u \times \bx_v\|^2}\,
 \begin{bmatrix}
  \mu \sprod{\br,\bn}  & \mu \sprod{\bs,\bn} \\
  \lambda \sprod{\br,\bn}  & \lambda \sprod{\bs,\bn}
 \end{bmatrix}
 + O(1/h)
 = O(1/h)
 .
\]
Its eigenvalues are the principal curvatures $\kappa_{1,2}$,
which are of the same order of magnitude,
\[
 \kappa_{1,2} = O(1/h).
\]
The surface element is $d\mu = \|\bx_u \times \bx_v\|\, dudv$,
where $\|\bx_u \times \bx_v\|= O(h)$. Hence,
there exists a constant $c$ such that 
$|\kappa_{1/2}|^p\, \|\bx_u \times \bx_v\|\le c h^{1-p}$, 
and we obtain\myHighlight{, using $h = \max\{u,v\}$,}
\[
 \int_{[0,H]^2} |\kappa_{1,2}|^p\, d\mu \le
 c \int_0^H\int_0^H h^{1-p}\, dudv
 =
 \frac{2cH^{3-p}}{3-p}
 < \infty
\]
for $p<3$, as claimed.
\end{proof}
In applications, surface patches are often given in B-spline format.
The following theorem specifies conditions for control points that
are equivalent to Definition~\ref{def:rounded}.
These conditions take the simplest form when the boundary knots
have maximal multiplicity so that we focus on that case.
In particular, B\'ezier patches are covered.
\begin{theorem}
\label{thm:spline}
Denote by $b^1_{j}$ and $b^2_{k}$ the B-splines of 
degrees $n_1,n_2 \ge 2$ with knots
\[
 T_1 = [\underbrace{0,\dots,0}_{n_1+1 {\rm\ 
times}},\tau_1^1,\tau_2^1,\dots,\tau^1_{N_1}]
 ,\quad
 T_2 = [\underbrace{0,\dots,0}_{n_2+1 {\rm\ 
times}},\tau_1^2,\tau_2^2,\dots,\tau^2_{N_2}]
 ,
\]
respectively, where $\tau^1_1,\tau^2_1>0$.
The spline surface
\[
 \bx(u,v) = \sum_{j= 0}^{N_1-1} \sum_{k=0}^{N_2-1}
 b^1_{j}(u) b^2_{k}(v)\, \bp_{j,k}
\]
with control points $\bp_{j,k} \in \R^3$ has a rounded corner at $(0,0)$
according to Definition~\ref{def:rounded}
if the following conditions are satisfied:
\begin{itemize}
\item
There exist weights $\alpha_1,\alpha_2 \in (0,1)$ with $\alpha_1+\alpha_2 = 1$
such that 
\begin{equation}
\label{eq:spline_anti}
 \bp_{0,0} = \alpha_1\bp_{1,0} + \alpha_2\bp_{0,1}
 .
\end{equation}
\item The three vectors
\begin{align}
  \br^* &:= (n_1-1) \tau^1_1 \alpha_1(\bp_{2,0}-\bp_{0,0}) +
  n_1\tau^1_2\alpha_2\,(\bp_{1,1}-\bp_{0,0}) \nonumber\\
  \bs^* &:= (n_2-1) \tau^2_1 \alpha_2(\bp_{0,2}-\bp_{0,0}) +
  n_2\tau^2_2\alpha_1(\bp_{1,1}-\bp_{0,0})\label{eq:spline_coplan}\\
  \bt^* &:= \bp_{1,0} - \bp_{0,1} \nonumber
\end{align}
are linearly dependent.
\item \myHighlight{We have} 
\begin{equation}
\label{eq:spline_positive}
 \sprod{\bt^* \times \br^*,\bt^* \times \bs^*} > 0 
  .
\end{equation}
\end{itemize}
\end{theorem}
\begin{proof}
In the following, $c_1,\dots,c_5$ denote real factors, the specific values of 
which are irrelevant.
The partial derivatives of $\bx$ are given by 
\begin{align*}
 \bxi_{0,0} &= \bp_{0,0}\\
 \bxi_{1,0} &= \frac{n_1}{\tau_1^1} (\bp_{1,0} - \bp_{0,0})\\
 \myHighlight{\bxi_{0,1}} & \myHighlight{= \frac{n_2}{\tau^2_1} (\bp_{0,1} - \bp_{0,0})}\\
 \bxi_{2,0} &= \frac{n_1(n_1-1)}{\tau^1_1\tau^1_2} 
 (\bp_{2,0} - \bp_{0,0}) +c_1 \bt^* \\
 \bxi_{1,1} &= \frac{n_1n_2}{\tau^1_1\tau^2_1}(\bp_{1,1} - \bp_{0,0}) + 
 c_2\bt^*\\
 \bxi_{0,2} &= \frac{n_2(n_2-1)}{\tau^2_1\tau^2_2}
 (\bp_{0,2} - \bp_{0,0}) + c_3\bt^* \\
\end{align*}
First, we observe that $\bt^* \neq \bZero$ because of 
\eqref{eq:spline_positive}. Further, by
\eqref{eq:spline_anti}, we have
$\bp_{1,0}-\bp_{0,0} = \alpha_2 \bt^*$ and 
$\bp_{0,1}-\bp_{0,0} = \alpha_1 \bt^*$. Hence, the condition of antiparallelism 
is satisfied with
\[
 \bt = \bt^*/\|\bt^*\|
 ,\quad 
 \lambda = n_1\alpha_2 \|\bt^*\|/\tau^1_1
 ,\quad
 \mu = n_2\alpha_1 \|\bt^*\|/\tau^2_1
 .
\]
Second, we find
\[
  \br = \frac{n_1 n_2\, \|\bt^*\|}{(\tau_1^1)^2\tau_1^2\tau_2^1}\, \br^*
  + c_4\bt^*
  ,\quad 
  \bs = \frac{n_1 n_2\, \|\bt^*\|}{(\tau_1^2)^2\tau_1^1\tau_2^2}\, \bs^*
  + c_5\bt^*
\]
and conclude that $\br,\bs,\bt$ are linearly dependent if so are 
$\br^*,\bs^*,\bt^*$, thus establishing coplanarity.
Third, positivity follows from 
\[
 \sprod{\bt\times \br,\bt \times \bs}
 =
 \frac{n_1^2 n_2^2}{(\tau^1_1\tau^2_1)^3\tau^1_2\tau^2_2}\,
 \sprod{\bt^*\times \br^*,\bt^* \times \bs^*} > 0
 .
\]
\end{proof}

\section{Experimental results}
\label{sec:experiments}

In this section, we apply the rounded corner constraints to B-spline models.
First, the impact of these conditions on the approximation of a 
hemisphere, its normal, and curvature is investigated. Then, they are utilized to 
improve the representation of a boat fender model. 

\subsection{Approximation of a hemisphere}
\label{sec:hemisphere}

The first example considers the approximation of a hemisphere with 
radius $r=1$ by a spline surface with four rounded corners. 
\Cref{fig:hemisphereTarget} illustrates the reference surface 
$\mathbf{y}$ and details its parametrization.
\begin{figure}[b!]
    \tikzfig{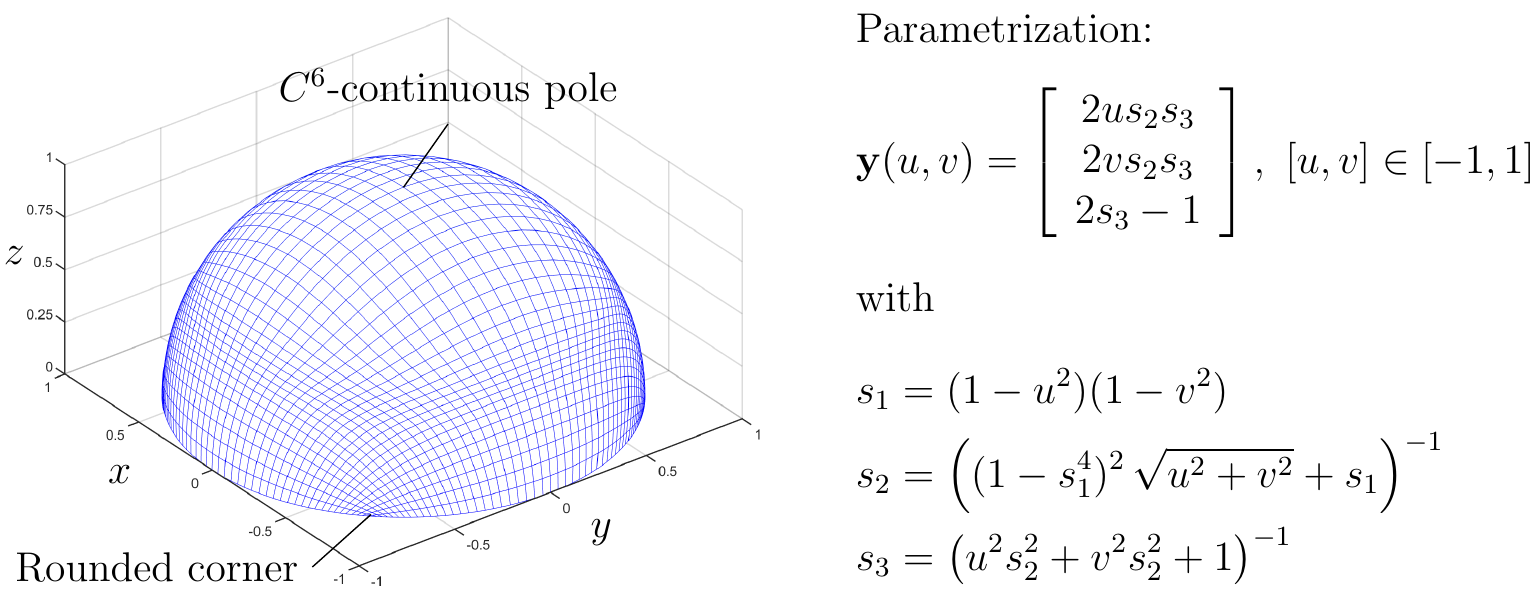}{0.195}{%
        Parametrization of the hemisphere over the domain 
        $[-1,1]^2$.
    }{fig:hemisphereTarget}
\end{figure}
%
Based on $\mathbf{y}$, we construct single-patch B-spline surfaces
$\bx : [-1,1]^2 \to \R^3$ of bi-degree $(\splineDegree,\splineDegree)$ with knot spacing
$h = 2^{-\ell}$ for various values of $\splineDegree$ and $\ell$.
The following schemes are employed:
\begin{itemize}
    \item \fitByL: approximation by conventional $L^2$-projection
    \item \fitByRC: $L^2$-projection including the rounded corner 
constraints (\myHighlight{RCC})
\end{itemize}
In both cases, the $L^2$-projection is performed in two steps: first, 
the boundary control points are fitted in the 
$xy$-plane, and subsequently, the inner control points are computed. 
This procedure yields better visual comparability of coarse discretization. 
For the \fitByRC case, we set $\alpha_1 = \alpha_2 = 1/2$ and use the 
known limit 
normal $\mathbf{n}$ of each rounded corner to specify the 
corresponding orientation of the limit tangent space $\mathbf{T}$. 
Lagrange multipliers are used to enforce these conditions together 
with the antiparallelism and coplanarity constraints, i.e., 
\cref{eq:spline_anti} and \cref{eq:spline_coplan}. 
After the construction, we check for onesidedness 
\cref{eq:spline_positive}. In our experiments, this condition was never violated.

The implementation is first validated by a convergence 
study of the approximation error. Therefore, the maximal error of different B-splines with various degrees 
and element numbers are summarized in \cref{fig:hemispereMaxErrApproxAndNormal}(a). 
\begin{figure}[t]
    \tikzfig{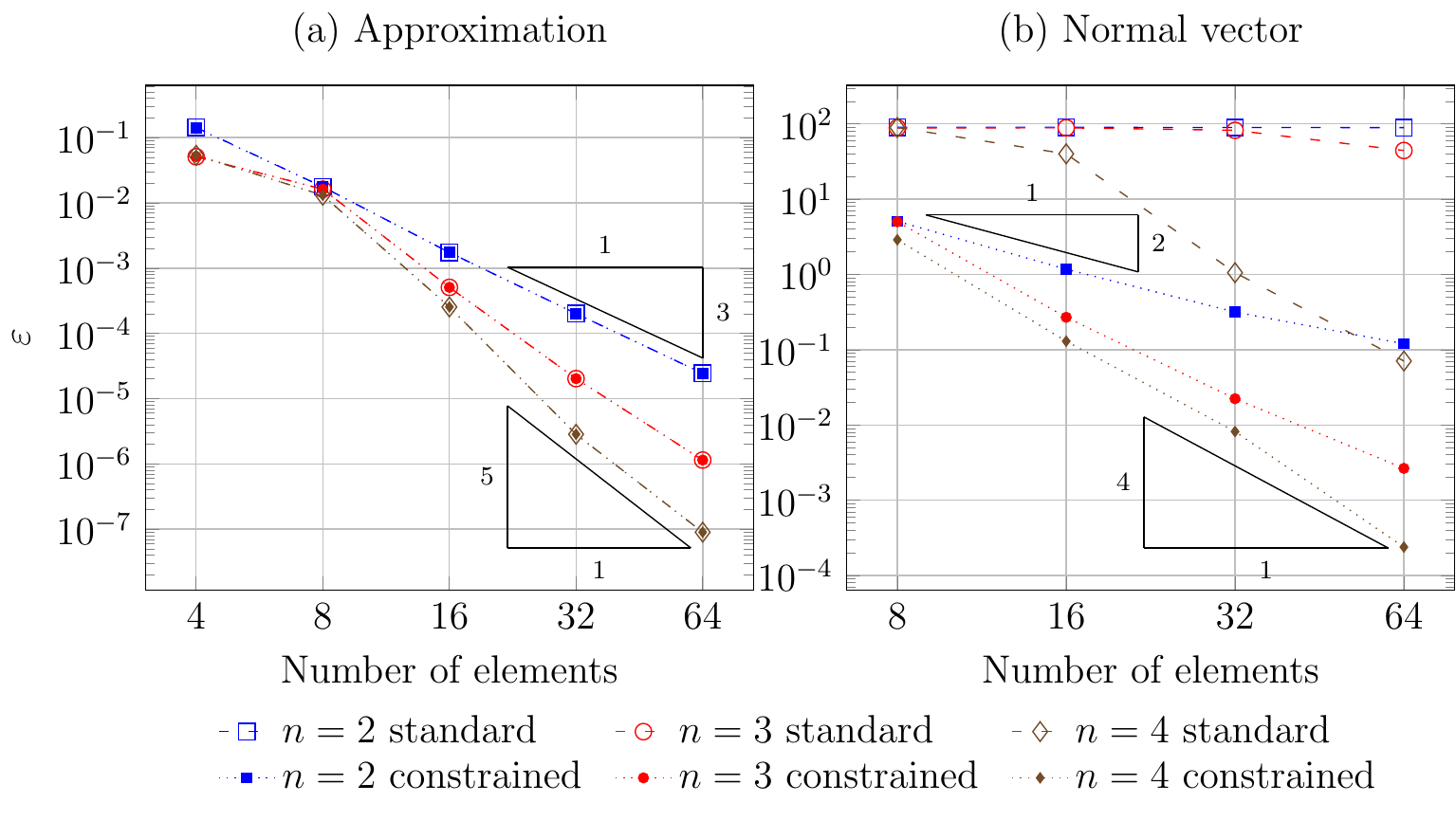}{0.95}{%
        Relative maximal approximation error \emph{(left)} and maximal deviation of normal vectors \emph{(right)} of both schemes for different polynomial degrees $\splineDegree$ in dependence of the number of elements per parametric direction.}
        {fig:hemispereMaxErrApproxAndNormal}
\end{figure}
Note that both schemes obtain optimal convergence rates. In fact, the graphs of the \fitByL and the \fitByRC approach 
are almost identical, indicating that the effects of the constraints on the approximation power are marginal.

Let us now focus on geometric aspects of the approximation 
process.
First, we investigate the error in representing normal 
vectors, which is measured by the angle between the 
reference normal and that of the approximation. 
\cref{fig:hemispereMaxErrApproxAndNormal}(b) shows the maximal deviations of normal vectors, again for different degrees and both schemes.
It is worth noting that we never evaluate directly in a rounded corner, where the normal may be undefined. 
Obeying the constraints for rounded corners yields convergence of normals at rates growing with the chosen 
degree, while standard approximation performs significantly worse.
Those issues of the \fitByL scheme are induced by the loss of regularity near rounded corners.
In \cref{fig:hemisphereNormalErrorAlongDiagonal}, the error in representing 
normal vectors is plotted along a diagonal 
emanating from a rounded corner of a B-spline surface 
$\splineFit(\splineU,\splineV)$ with $\splineDegree=3$ and $\ell=3$.
To be precise, the errors are evaluated at $\splineU=\splineV=\alpha$ 
with $\alpha\in[10^{-7},0.1]$.
Note that with the proposed RCC, the error in the normal vector goes 
to zero, as $\alpha\rightarrow0$.
\begin{figure}[h!]
    \tikzfig{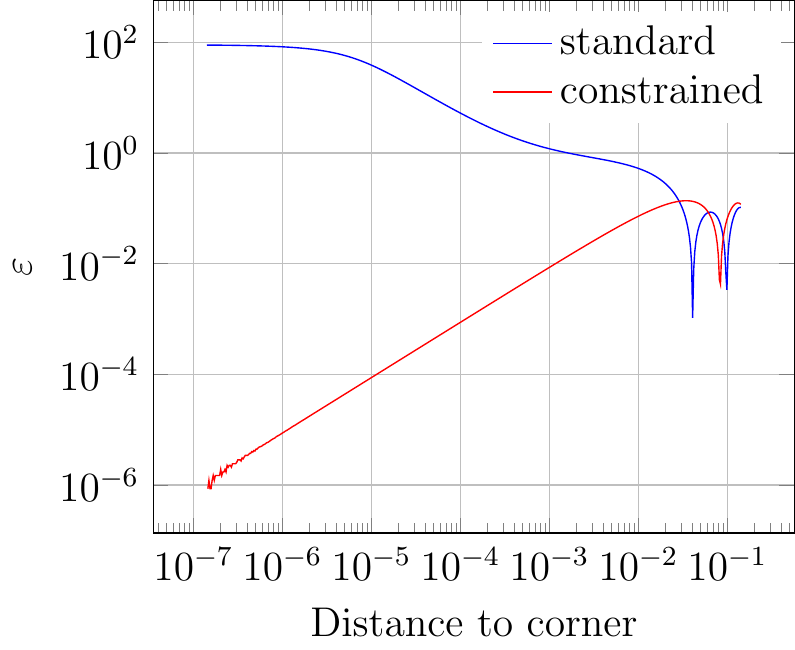}{0.95}{%
        Error distribution along a diagonal emitting from a rounded 
corner for a B-spline surface with degree $\splineDegree=3$ and $\ell=3$.}{fig:hemisphereNormalErrorAlongDiagonal}
\end{figure}
\begin{figure}[b!]
    \tikzfig{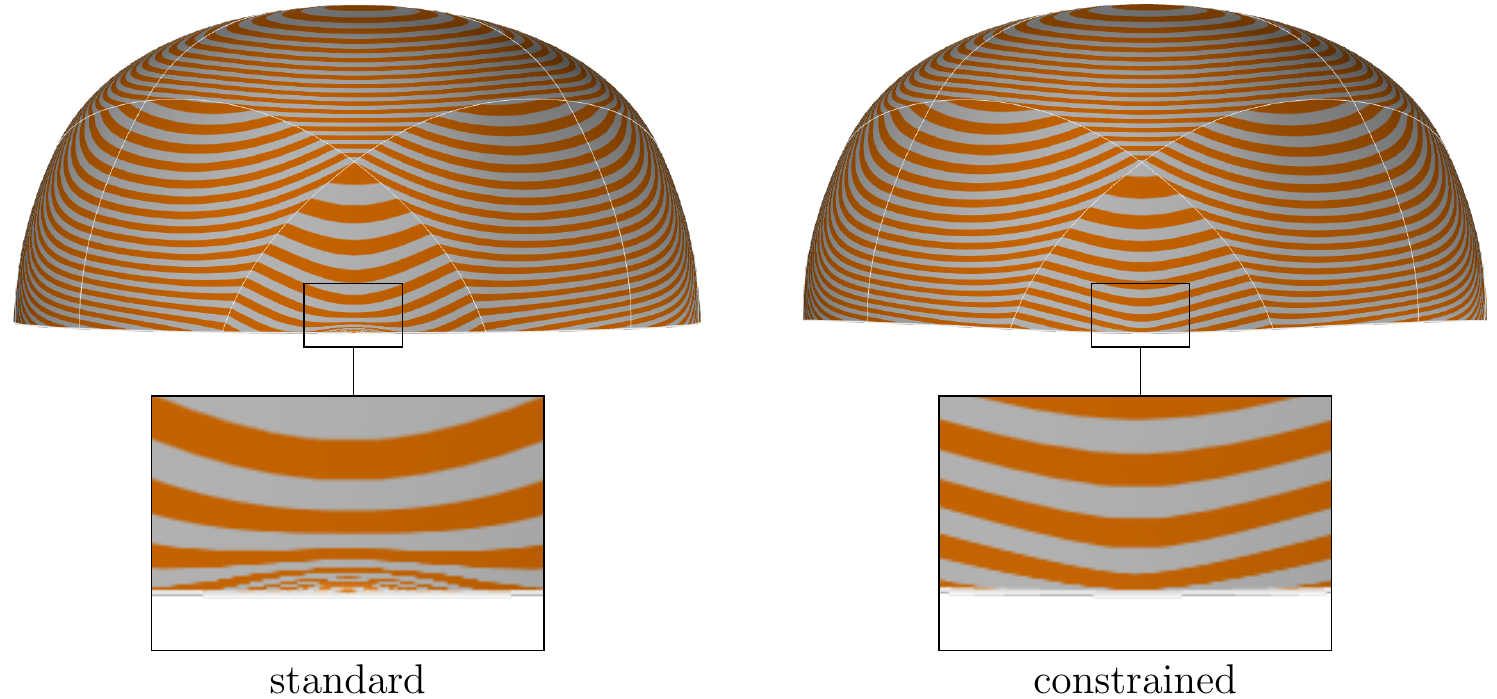}{0.2}{%
        Reflection lines of a B-spline with $\splineDegree=2$ and  $\ell=2$. The white lines are the $C^1$ isolines of the surfaces 
and the close-ups show the reflections around a rounded corner.
    }{fig:hemisphereZebraP2}
\end{figure}

Finally, we demonstrate the impact of rounded corners on curvature.  
Therefore, \cref{fig:hemisphereZebraP2} 
shows the reflection lines of each approximation scheme for 
surfaces $\splineFit$ with $\splineDegree=2$ and $\ell=2$.
Note the irregularities close to the rounded corner in the \fitByL approach, which vanish in the \fitByRC case.

\subsection{Watertight boat fender model}

The following example utilizes the \fitByRC $L^2$-projection investigated in \cref{sec:hemisphere} in the 
context of a modeling process.
In particular, we consider a ``watertight'' boat fender B-spline 
model, i.e., the boundary representation has no trimmed patches, and 
all splines surfaces are connected by explicit continuity conditions.
The initial model is constructed by watertight Boolean operations 
detailed, in  \cite{Urick19}.
These operations connect intersecting surfaces in a non-trimmed 
$C^0$-continuous manner. At the same time, this construction may 
introduce rounded corners in the spline model.
\Cref{fig:boatfenderRenderingCloseUpInitial} illustrates the initial 
model of the boat fender. 
Note that the close-up shows four rendering defects. There the 
model possesses rounded corners. 
\begin{figure}[t]
    \tikzfig{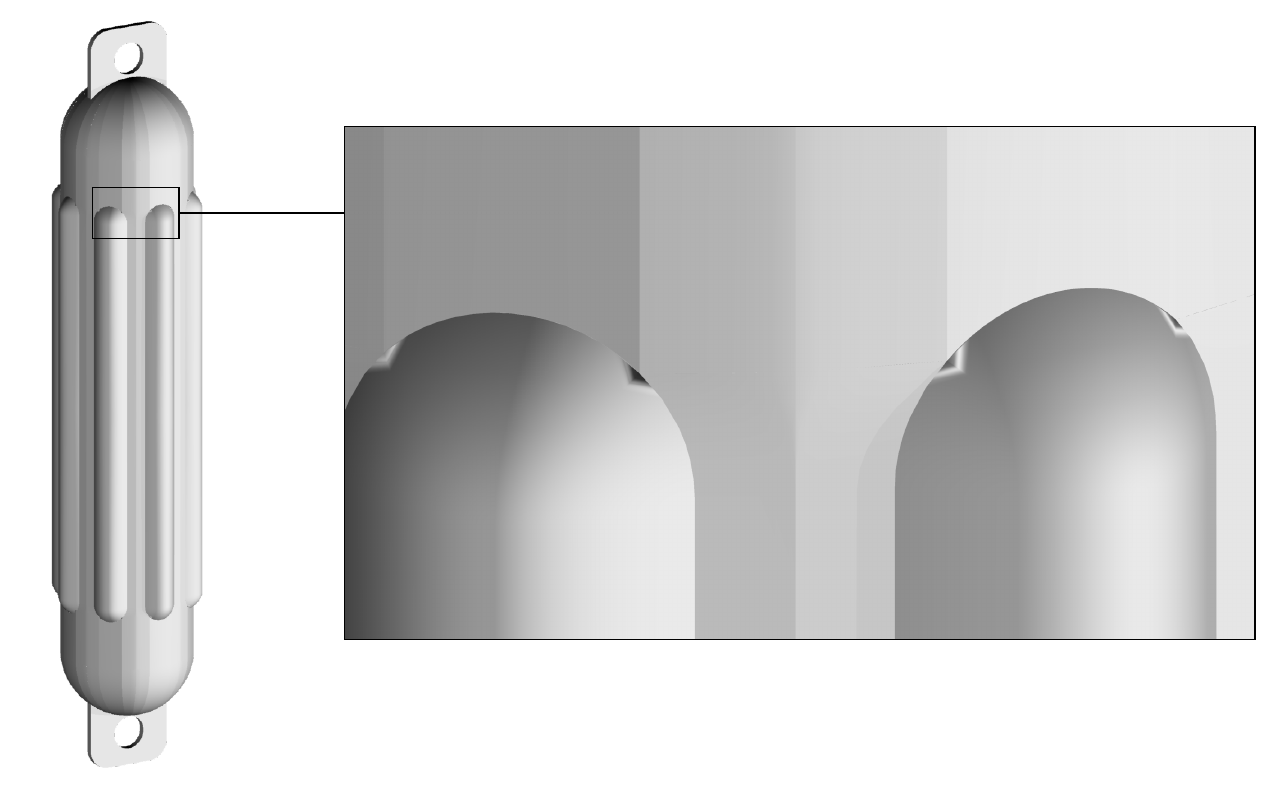}{1.5}{%
        Initial watertight model of the boat fender with rendering 
        defects due to rounded 
corners.}{fig:boatfenderRenderingCloseUpInitial}
\end{figure}
\begin{figure}[!b]
    \tikzfig{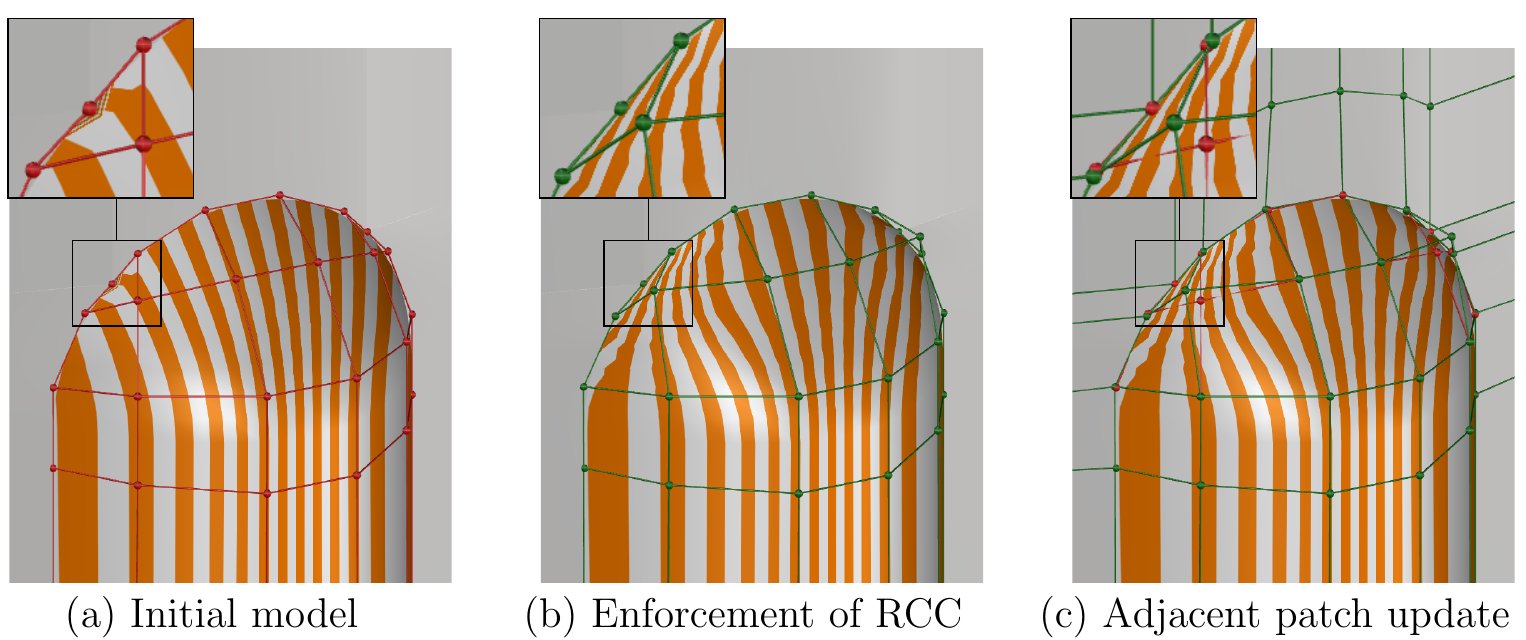}{0.18}{%
        Implementation of rounded corner constraints (RCC): 
        (a) detection of a surface patch with rounded corners, 
        (b) adjustment of the control points of this patch by 
         \fitByRC $L^2$-projection, and (c) update of the control 
         points of the surface patches adjacent to the rounded 
        corner. 
        The orange reflection lines indicate the impact of the 
constraints. 
        Initial control points are shown in red, while the final ones 
are shown in green. In (c), they are  plotted on top of each other for 
better comparison.
    }{fig:boatfenderRenderingCloseUpRCCSetUp}
\end{figure}

\Cref{fig:boatfenderRenderingCloseUpRCCSetUp} outlines how the 
rounded corner constraints can be used to improve watertight models:
\begin{itemize}
    \item[(a)] detect all rounded corners and \myHighlight{the adjacent} surfaces,
    \item[(b)] employ the \fitByRC approximation scheme described in \cref{sec:hemisphere}, and
    \item[(c)] update the adjacent surfaces to maintain a watertight representation.
\end{itemize}
\myHighlight{Here, adjacent surfaces refer to surfaces that are connected to another surface's rounded corner. Their update is necessary since the \fitByRC approximation scheme affects the control points along the shared surface edge (cf.~red and green control points in \cref{fig:boatfenderRenderingCloseUpRCCSetUp}(c)).}
The resulting model is shown in \cref{fig:boatfenderRenderingCloseUpRCC}.

\begin{figure}[htb]
    \tikzfig{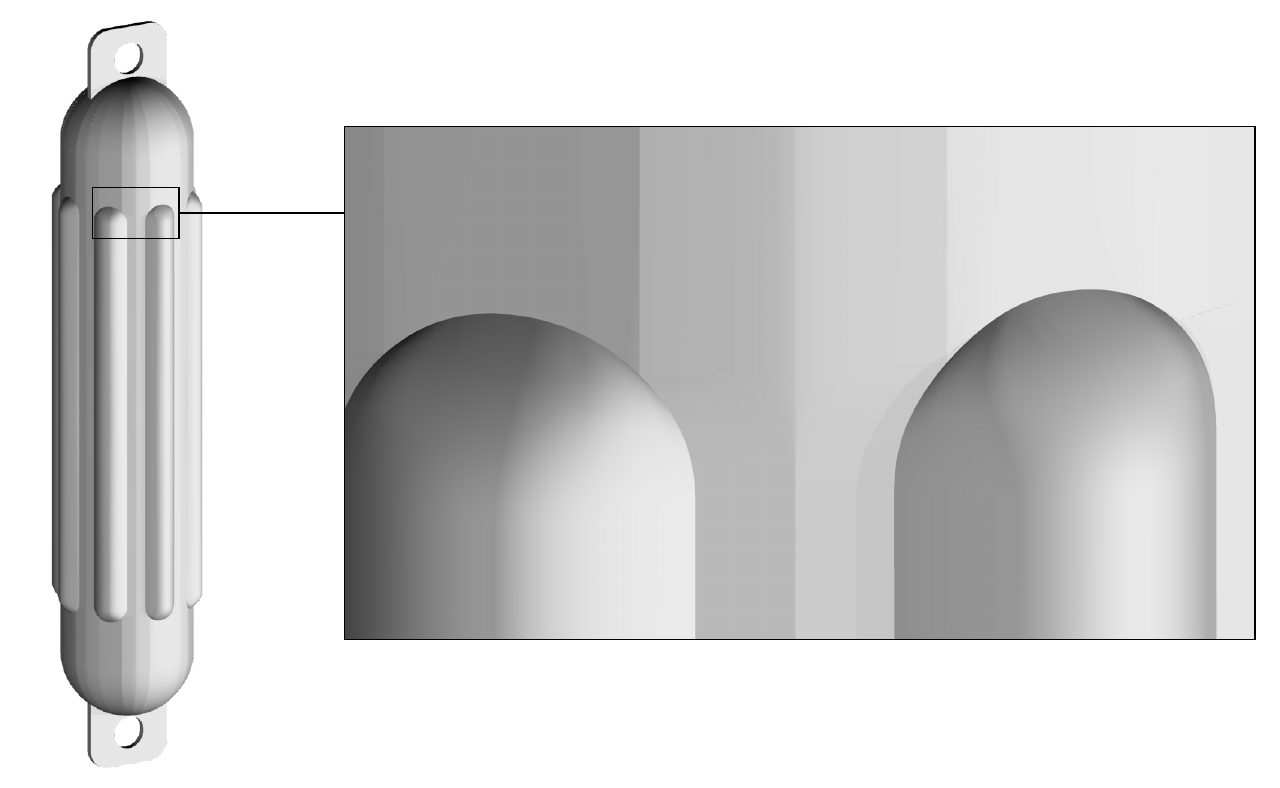}{1.5}{%
        Final watertight model of the boat fender with rounded corner constraints
    }{fig:boatfenderRenderingCloseUpRCC}
\end{figure}

\clearpage
\section{Acknowledgment}
The work of Benjamin Marussig was partially supported by the Austrian Research Promotion Agency (FFG, project 883886), and the joint DFG/FWF Collaborative Research Centre CREATOR (CRC/TRR 361, F90) at TU Darmstadt, TU Graz and JKU Linz.

\bibliographystyle{alpha}
\bibliography{references}

\end{document}